\documentclass[10pt]{article}
\usepackage{amssymb,amsmath,amsthm,mathrsfs,xspace,multicol}
\usepackage{amscd}
\usepackage[all]{xy}
\usepackage[breaklinks=true]{hyperref}

\newcommand{\comment}[1]{}
\newcommand{\blankbrac}{[ \cdot,\cdot ]}
\newcommand{\brac}[2]{[ #1,#2 ]}
\newcommand{\Hblank}{\lbrace \cdot,\cdot \rbrace_{\mathrm{h}}}
\newcommand{\Sblank}{\lbrace \cdot,\cdot \rbrace_{\mathrm{s}}}
\newcommand{\Hbrac}[2]{\left \lbrace #1,#2 \right\rbrace_{\mathrm{h}}}
\newcommand{\Sbrac}[2]{\left \lbrace #1,#2 \right\rbrace_{\mathrm{s}}}
\newcommand{\h}{\mathrm{h}}
\newcommand{\s}{\mathrm{s}}
\newcommand{\ham}{\mathrm{Ham}\left(X\right)}
\newcommand{\cinf}{C^\infty(X)}
\newcommand{\lie}[2]{L_{v_{#1}}\, #2}
\newcommand{\Lie}[2]{L_{v_{#1}}\,#2}
\newcommand{\U}{{\rm U}}
\newcommand{\R}{\mathbb{R}}
\newcommand{\Z}{\mathbb{Z}}
\newcommand{\ip}[1]{\iota_{v_{#1}}}
\newcommand{\ST}{\mathbb{R}^{1,d-1}}
\newcommand{\CYL}{\mathbb{R} \times S^{1}}
\newcommand{\vol}{{\rm vol}}

\newcommand{\maps}{\colon}    

\theoremstyle{plain}
\newtheorem{theorem}{Theorem}[section]
\newtheorem{prop}[theorem]{Proposition}

\newtheorem{definition}[theorem]{Definition}

\theoremstyle{definition}
\newtheorem{example}[theorem]{Example}

\title{Categorified Symplectic Geometry \\ and the Classical String \\}
\author{John C.\ Baez, Alexander E.\ Hoffnung, Christopher L.\ Rogers \\ 
\\
Department of Mathematics, 
University of California, 
\\ Riverside, California 92521, USA}

\date{August 2, 2008}

\begin{document}

\maketitle

\begin{abstract}
\noindent
A Lie 2-algebra is a `categorified' version of a Lie algebra: that is,
a category equipped with structures analogous those of a Lie algebra,
for which the usual laws hold up to isomorphism.  In the classical
mechanics of point particles, the phase space is often a symplectic
manifold, and the Poisson bracket of functions on this space gives a
Lie algebra of observables.  Multisymplectic geometry describes an
$n$-dimensional field theory using a phase space that is an 
`$n$-plectic manifold': a finite-dimensional manifold equipped with
a closed nondegenerate $(n+1)$-form.  Here we consider the case $n = 2$.  
For any 2-plectic
manifold, we construct a Lie 2-algebra of observables.  We then
explain how this Lie 2-algebra can be used to describe the dynamics of
a classical bosonic string.  Just as the presence of an
electromagnetic field affects the symplectic structure for a charged
point particle, the presence of a $B$ field affects the 2-plectic
structure for the string.
\end{abstract}

\section{Introduction}
\label{introduction}

It is becoming clear that string theory can be viewed as a
`categorification' of particle physics, in which familiar algebraic
and geometrical structures based in set theory are replaced by their
category--theoretic analogues.  The basic idea is simple.  While a
classical particle has a position nicely modelled by an element 
of a set, namely a point in space:
\[          \bullet  \]
the position of a classical string is better modelled by a morphism in a 
category, namely an unparametrized path in space:
\[
\xymatrix{
   \bullet \ar@/^1pc/[rr]^{}
&& \bullet
}
\]
Similarly, while particles have worldlines in spacetime,
which can be thought of as morphisms, strings have worldsheets, 
which can be thought of as 2-morphisms.

So far this viewpoint has been most fruitful in studying the
string-theoretic generalizations of gauge theory 
\cite{BaezSchreiber:2005}.   The first clue was the $B$ field 
in string theory.  The electromagnetic field contributes to the
change in phase of a charged particle as it moves through spacetime.  
This field is locally described by a 1-form $A$, which we integrate 
along the particle's worldline to compute a phase change.
The $B$ field contributes to the phase change of a charged
string in a similar way: it is locally described by a 2-form,
which we integrate over a string's worldsheet.  
When we seek a global description suitable for nontrivial spacetime
topologies, the electromagnetic field is better thought of 
as a connection on a $\U(1)$ bundle.  Similarly, the $B$-field 
is globally described by a connection on the categorified version of 
a $\U(1)$ bundle, namely a $\U(1)$ gerbe 
\cite{Brylinski,Freed:2001,Freed-Witten,Zunger}.   

Later it was found that connections on nonabelian gerbes also play a
role in string theory \cite{AschieriCantiniJurco:2003,
AschieriJurco:2004, BreenMessing:2001}.  Nonabelian gerbes are a
special case of 2-bundles: that is, bundles with a smooth category
rather than smooth manifold as fiber \cite{Bartels:2004}.  To
understand connections on general 2-bundles, it was necessary to
categorify the concepts of Lie group and Lie algebra, obtaining the
notions of Lie 2-group \cite{BCSS,BaezLauda:2003} and Lie 2-algebra
\cite{BaezCrans:2004, Roytenberg}.

Still more recently, {\it iterated} categorification has become
important in understanding the generalizations of gauge theory
suitable for higher-dimensional membranes
\cite{Diaconescu-Moore-Freed,SSS,SSS2}.  It is clear by now that to
understand the behavior of these membranes, we need to
study $n$-connections on $n$-bundles: that is, structures analogous to
connections that live on something like a bundle with a smooth
$n$-category as fiber.  In the very simplest case --- a topologically
trivial $n$-bundle with the simplest nontrivial abelian `$n$-group'
playing the role of gauge group --- an $n$-connection is just an
$n$-form on the base space.  In a straightforward generalization of
electromagnetism, the integral of this $n$-form over the membrane's
`worldvolume' contributes to its change in phase. 

Given all this, we should expect that as we look deeper into the 
analogy between point particles, strings, and higher-dimensional
membranes, we should find more examples of categorification.  
Perhaps the most obvious place to look is symplectic geometry.
The reason is that symplectic geometry {\it also} uses a connection
on a $\U(1)$ bundle to describe the change of phase of a point particle.  

The simplest example is a free particle moving in some Lorentzian
manifold $M$ representing spacetime.  If we keep track of the 
particle's momentum as well as its position, it traces out a path 
in the cotangent bundle $X = T^* M$.  The cotangent bundle is equipped 
with a canonical 1-form $\alpha$, and we can integrate $\alpha$ 
over this path to determine the particle's change of phase.  
This is not the historical reason why $X$ is called a `phase 
space', but the coincidence is a happy one.

The exterior derivative $\omega = d\alpha$ plays an important role in
this story.  First, by Stokes' theorem, the integral of this 2-form
over any disc in $X$ measures the change of phase of a particle as it
moves around the boundary of the disc.  A deeper fact is that $\omega$
is a symplectic structure: that is, not only closed but nondegenerate.
This lets us take any smooth function $F \maps X \to \R$ and find a
unique vector field $v_F$ such that
\[                   \ip{F} \omega = -dF  \]
where $\iota$ stands for interior product.  We should think of $F$ as an 
`observable' assigning a number to any state of the particle.  In good 
situations, the vector field $v_F$ will generate a one-parameter group 
of symmetries of $X$: that is, a flow preserving the symplectic structure 
$\omega$.  So, the symplectic nature of $\omega$ guarantees that {\it 
observables give rise to symmetries}.  Moreover, by measuring how rapidly one
observable changes under the one-parameter group of symmetries
generated by another, we obtain a binary operation on observables, the
Poisson bracket:
\[           \{F,G\} = L_{v_F} G  \]
where $L$ stands for Lie derivative.  This makes the vector space of
observables into a Lie algebra.  

Symplectic geometry generalizes this idea by replacing $T^* M$
with a more general phase space $X$.  We could simply let 
$X$ be any manifold equipped with a 1-form $\alpha$ such that 
$\omega = d\alpha$ is symplectic.  However, a 1-form is the same as a 
connection on a trivial $\U(1)$ bundle, and $\omega$ is then the
curvature of this connection.  Since physics is local, it makes 
more sense to equip $X$ with a {\it locally} trivial $\U(1)$ bundle 
$P \to X$, together with a connection on $P$ whose curvature 2-form 
$\omega$ is symplectic.  This is the basic context for geometric 
quantization.

We can study symplectic geometry without assuming that the symplectic
2-form $\omega$ is the curvature of a connection on some $\U(1)$
bundle.  In particular, we still obtain a Lie algebra of observables
using the formulas above.  But some of the physical meaning of the
symplectic structure only reveals itself in the presence of a $\U(1)$
bundle: namely, that the integral of $\omega$ over any disc in $X$
measures the change of phase of a particle as it moves around the
boundary of this disc.  So, in geometric quantization the $\U(1)$
bundle is crucial.  We can build such a bundle whenever we can 
lift the de Rham cohomology class $[\omega] \in H^2(X,\R)$ to 
an element of the integral cohomology $H^2(X,\Z)$.

Now let us consider how all this generalizes when we move from point
particles to strings.  As a first step towards understanding this, let
us return to the point particle moving in a spacetime manifold $M$.
We have said that the particle's phase changes in a way described by
integrating the canonical 1-form $\alpha$ along its path in $T^* M$.
However, in the presence of the electromagnetic field there is an
additional phase change due to electromagnetism, at least when the
particle is charged.  To take this into account, we add to $\alpha$
the 1-form $A$ describing the electromagnetic field, pulled back from
$M$ to $T^* M$.  We then redefine the symplectic structure to be
$\tilde{\omega} = d(\alpha + A)$.  So, {\it electromagnetism affects the
symplectic structure on the cotangent bundle of spacetime}.  A more
detailed account of this can be found in the book by Guillemin and
Sternberg \cite{GS}.

This suggests that when we pass from point particles to strings, and 
the electromagnetic field is replaced by the $B$ field, we should
correspondingly adjust our concept of `symplectic structure'.  Instead 
of a canonical 1-form, we should have some sort of canonical 2-form
on phase space, so we can add the $B$ field to this 2-form.    But 
this in turn suggests that the analogue of the symplectic structure 
will be a 3-form!

This raises the puzzle: {\it how can we generalize symplectic geometry
with a 3-form replacing the usual 2-form?}

Amusingly, the answer is very old: it goes back to the work of DeDonder
\cite{DeDonder} and Weyl \cite{Weyl} in the 1930s.  Their ideas have
been more fully developed in the subject called `multisymplectic
geometry'.  For an introduction, try for example the papers by Gotay,
Isenberg, Marsden and Montgomery \cite{GIMM}, H\'elein and Kouneiher
\cite{Helein,Helein-Kouneiher}, Kijowski \cite{Kijowski}, and Rovelli
\cite{Rovelli}.  In particular, Gotay {\it et al} have already applied
multisymplectic geometry to classical string theory.  There are
various ways to do this.  In this introduction we take a very naive
approach, which will be corrected in Section \ref{multisymplectic}.

To begin with, note that just as the position and velocity of 
a point particle in the spacetime $M$ are given by a point in
the tangent bundle $TM$, we could try to describe the position and 
velocity of a string by a point in $\Lambda^2 TM$ --- that is, a point 
in $M$ together with a tangent {\it bivector}.  Similarly, just as
the position and momentum of a particle are given by a point in 
$T^* M$, we could try to describe the position and momentum of a 
string by a point in $\Lambda^2 T^*M$.  

Just as $T^* M$ is equipped with a canonical 1-form, the generalized
phase space $X = \Lambda^2 T^* M$ is equipped with a canonical 2-form
$\alpha$, as described in Example \ref{Poincare-Cartan} below.  The
corresponding 3-form $\omega = d \alpha$ is `multisymplectic', meaning
that it is closed and also nondegenerate in the following sense:
\[             \iota_v \omega = 0 \Rightarrow v = 0  \]
for all vector fields $v$.   This means that for any 1-form 
$F$, there is at most one vector field $v_F$ such that
\[                   \ip{F} \omega = -dF . \]
This resembles the equation we have already seen in symplectic
geometry, which associates symmetries to observables.  But there is a
difference: now $v_F$ may not exist.  So, we should consider
a 1-form $F$ on $X$ to be an observable only when there exists a
vector field $v_F$ satisfying the above equation.

We can then define a Poisson bracket of observables by the usual
formula:
\[           \{F,G\} = L_{v_F} G . \]
The result is always another observable.  But, we do not obtain
a Lie algebra of observables, because this Poisson bracket is only
antisymmetric {\it up to an exact 1-form}.  Exact 1-forms are
always observables, but they play a special role, since they 
give rise to trivial symmetries: if $F$ is exact, $v_F = 0$.   

This suggests that in the stringy analogue of symplectic geometry
we should seek, not a Lie algebra of observables, but a 
{\it Lie 2-algebra} of observables --- that is, a category 
resembling a Lie algebra, with observables as objects.  In this
category two observables $F$ and $G$ will be deemed `isomorphic' 
if they differ by an exact 1-form.  This guarantees that they generate
the same symmetries: $v_F = v_G$.

Indeed, such a Lie 2-algebra exists.  After reviewing multisymplectic
geometry in Section \ref{multisymplectic}, we prove
in Thm.\ \ref{hemistrict} that for any manifold $X$ equipped 
with a closed nondegenerate 3-form $\omega$, there is a Lie 2-algebra 
for which:
\begin{itemize}
\item An object is a 1-form $F$ on $X$ for which there exists
a vector field $v_F$ with  $\ip{F} \omega = -dF$.
\item A morphism $f \maps F \to F'$ is a function $f$ such that $F + df = F'$.
\item The bracket of objects $F,G$ is $L_{v_F} G$.  
\end{itemize}

On a more technical note, this Lie 2-algebra is `hemistrict' in the
sense of Roytenberg \cite{Roytenberg}.  This means that the Jacobi
identity holds on the nose, but the skew-symmetry of the bracket holds
only up to isomorphism.  In Thm.\ \ref{semistrict} we construct
another Lie 2-algebra with the same objects and morphisms, where the
Lie bracket of observables is given instead by $\ip{G} \ip{F} \omega$.
This Lie 2-algebra is `semistrict', meaning that the bracket is
skew-symmetric, but the Jacobi identity holds only up to isomorphism.  
In Thm.\ \ref{isomorphism} we show that these two Lie 2-algebras are
isomorphic.  This may seem surprising at first, but the notion of
`isomorphism' for Lie 2-algebra is sufficiently supple that
superficially different Lie 2-algebras --- one hemistrict, one
semistrict --- can be isomorphic.

In Section \ref{strings}, we apply these ideas to the classical
bosonic string propagating in Minkowski spacetime.  Following standard
ideas in multisymplectic geometry, we replace $\Lambda^2 T^* M$ with a
more sophisticated 2-plectic manifold: the first cojet bundle of the
bundle $\Sigma \times M \to \Sigma$, where $\Sigma$ is a surface
parametrizing the string worldsheet.  We explain how to derive the
equations of motion for the string from a 2-plectic formulation
involving this phase space.  We describe an observable 1-form $H$ on
this phase space whose corresponding vector field $v_H$ generates time
evolution.  We also describe how the presence of a $B$ field modifies
the 2-plectic structure.  Finally, we list some open questions in
Section \ref{conclusions}.

\section{Multisymplectic Geometry}
\label{multisymplectic}

The idea of multisymplectic geometry is simple and beautiful:
associated to any $n$-dimensional classical field theory there 
is a finite-dimensional `extended phase space' $X$ equipped 
with a nondegenerate closed $(n+1)$-form $\omega$. When $n = 1$,
we are back to the classical mechanics of point particles and 
ordinary symplectic geometry.  When $n = 2$, the examples include
classical bosonic string theory, as explained in Section \ref{strings}.

However, at this point an annoying terminological question intrudes:
what do we call multisymplectic geometry for a fixed value of 
$n$?  The obvious choice is `$n$-symplectic geometry', but
unfortunately, this term already means something else \cite{Cartin}.  
So, until a better choice comes along, we will use the term 
`$n$-plectic geometry':

\begin{definition}
An $(n+1)$-form $\omega$ on a $C^\infty$ manifold $X$ is 
{\bf multisymplectic}, or more specifically
an \textbf{n}-{\bf plectic structure}, if it is both closed:
\[
    d\omega=0,
\]
and nondegenerate:
\[
    \forall v \in T_{x}X,\ \iota_{v} \omega =0 \Rightarrow v =0
\]
where we use $\iota_v \omega$ to stand for the interior product
$\omega(v, \cdot, \dots, \cdot)$.  
If $\omega$ is an $n$-plectic form on $X$ we call the pair $(X,\omega)$ 
a {\bf multisymplectic manifold}, or \textbf{n}-{\bf plectic manifold}.
\end{definition}

The references already provided contain many examples of
multisymplectic manifolds.  More examples, together with constraints
on which manifolds can admit $n$-plectic structures, have been
discussed by Cantrijn \textit{et al} \cite{Cantrijn:1999} and Ibort
\cite{Ibort:2000}.  Here we give four well-known examples.  

The first example arises in work related to the Wess--Zumino--Witten
model and loop groups:

\begin{example}
\label{lie.group}
If $G$ is a compact simple Lie group, there is a 
3-form $\omega$ on $G$ that is invariant under both left and 
right translations, which is unique up to rescaling.  It is given by
\[       \omega(v_1, v_2, v_3) = \langle v_1, [v_2, v_3] \rangle \]
when $v_i$ are tangent vectors at the identity of $G$ 
(that is, elements of the Lie algebra), and $\langle \cdot, \cdot \rangle$
is the Killing form.   This makes $(G,\omega)$ into a 2-plectic manifold.  
\end{example}

The second was already mentioned in the Introduction:

\begin{example}
\label{Poincare-Cartan}
Suppose $M$ is a smooth manifold, and let $X = \Lambda^n T^\ast M$ be the 
$n$th exterior power of the cotangent bundle of $M$.  Then there
is a canonical $n$-form $\alpha$ on $X$ given as follows: 
\[  \alpha(v_1, \dots, v_n) = x(d\pi(v_1), \dots, d\pi(v_n))  \]
where $v_1, \dots v_n$ are tangent vectors at the point $x \in X$, and
$\pi \maps X \to M$ is the projection from the bundle $X$ to the base
space $M$.  Note that in this formula we are applying the $n$-form $x
\in \Lambda^n T^\ast M$ to the $n$-tuple of tangent vectors
$d\pi(v_i)$ at the point $\pi(x)$.  The $(n+1)$-form
\[           \omega = d\alpha   \]
is $n$-plectic.  

Indeed, this can be seen by explicit computation.  Let $q^1, \dots,
q^d$ be coordinates on an open set $U \subseteq M$.  Then there is a basis
of $n$-forms on $U$ given by $dq^I = dq^{i_1} \wedge \cdots \wedge dq^{i_n}$
where $I = (i_1, \dots, i_n)$ ranges over multi-indices of length $n$.
Corresponding to these $n$-forms there are fiber coordinates $p_I$ which
combined with the coordinates $q^i$ pulled back from the base give a
coordinate system on $\Lambda^n T^* U$.  In these coordinates we have
\[    \alpha = p_I dq^I , \]
where we follow the Einstein summation convention to sum over
repeated multi-indices of length $n$.  It follows
that
\[    \omega = dp_I \wedge dq^I  .\]
Using this formula one can check that $\omega$ is indeed $n$-plectic.
\end{example}

The next example, involving an $n$-plectic manifold called
$\Lambda_1^n T^* E$, may seem like a technical variation on the theme
of the previous one.  However, it is actually quite significant, since
$n$-plectic manifolds of this sort serve as the extended phase spaces
for many classical field theories \cite{Carinena-Crampin-Ibort, GIMM,
Rovelli}.  In Section \ref{strings}, we use a 2-plectic manifold of
this sort as the extended phase space for the classical bosonic
string.

\begin{example}
\label{horizontal}
Let $\pi \maps E \to \Sigma$ be a fiber bundle over an $n$-dimensional
manifold $\Sigma$.  Given a point $y \in E$, a tangent vector $v \in T_y E$ 
is said to be {\bf vertical} if $d \pi(v) = 0$.  
There is a vector sub-bundle $\Lambda^{n}_{1}T^{\ast}E$ of the 
$n$-form bundle $\Lambda^n T^\ast E$ whose fiber at $y \in E$ 
consists of all $\beta \in \Lambda^n T^\ast_y E$ such that
\[          \ip{1}\ip{2} \beta =0   \]
for all vertical vectors $v_1,v_2 \in T_yE$.  
Let $i\maps \Lambda^{n}_{1}T^* E \hookrightarrow \Lambda^n T^* E$ 
denote the inclusion.  Let $\omega = d\alpha$ be the $n$-plectic
form defined in Example \ref{Poincare-Cartan}.  Then the pullback 
$i^* \omega$ is an $n$-plectic form on $\Lambda_1^n T^\ast E$.

Again, this can be seen by explicit calculation.  In our application 
to strings, $E$ will be a trivial bundle
$E = \Sigma \times M$ over $\Sigma$, and $\Sigma$ will be
equipped with a volume form.  It is enough to consider this case,
because proving that $i^* \omega$ is $n$-plectic is a local calculation,
and we can always trivialize $E$ and equip $\Sigma$ with a volume
form {\it locally}.

Let $q^1, \dots, q^n$ be local coordinates on $\Sigma$ and let
$u^1, \dots , u^d$ be local coordinates on $M$.  Then $\Lambda_1^n T^\ast E$
has a local basis of sections given by $n$-forms of two types: first, the 
wedge product of all $n$ cotangent vectors of type $dq^i$:
\[    dq^1 \wedge \cdots \wedge dq^n  \]
and second, wedge products of $n-1$ cotangent vectors of type $dq^i$ 
and a single one of type $du^a$:
\[   
dq^1 \wedge \cdots \wedge \widehat{dq^i} \wedge \cdots \wedge dq^n \wedge du^a . \]
Here the hat means that we omit the factor of $dq^i$.  
If $y = (x,u) \in \Sigma \times M$, this basis gives an isomorphism
\[   \Lambda^n_1 T^*_y E \;\; \cong \;\; \Lambda^n T^*_x \Sigma \; \oplus \;
\Lambda^{n-1} T^*_x \Sigma \otimes T^*_u M .\]

In calculations to come, it will be better to use the pulled back
volume form $\pi^* \vol$ as a substitute for the coordinate-dependent 
$n$-form $dq^1 \wedge \cdots \wedge dq^n$ on $E$.  This gives another basis of
sections of $\Lambda_1^n T^* E$, which by abuse of notation we call
\[  dQ = \pi^* \vol   \]
and
\[  dQ_i^a = 
\left(\pi^* \iota_{\partial/\partial q^i} \vol\right) \wedge du^a .\]
Corresponding to this basis there are
local coordinates $P$ and $P^i_a$ on $\Lambda_1^n T^* E$, which
combined with the coordinates $q^i$ and $u^a$ pulled back from $E$
give a local coordinate system on $\Lambda_1^n T^* E$.  In these
coordinates we have:
\begin{eqnarray}
\label{canonical.form}
          i^* \alpha = P dQ + P^i_a dQ_i^a,  
\end{eqnarray}
where again we use the Einstein summation convention.  It follows
that
\begin{eqnarray} 
\label{multisymplectic.form}
         i^* \omega = dP \wedge dQ + dP^i_a \wedge dQ_i^a  .
\end{eqnarray}
Using this formula one can check that $i^* \omega$ is indeed $n$-plectic.
\end{example} 

The manifold $\Lambda_1^n T^\ast E$ may seem rather mysterious, but
the next example shows that under good conditions it is isomorphic to
the `first cojet bundle' $J^1 E^\star$.  A point in the first jet bundle
$J^1 E$ records the value and first derivative of a section of $E$ at
some point of the base space $\Sigma$.  So, a first-order Lagrangian $\ell$ 
for a field theory where fields are sections of $E$ is a function $\ell
\maps J^1 E \to \R$.  $J^1 E$ is thus the natural home for the
Lagrangian approach to such field theories.  Similarly, the first
cojet bundle $J^1 E^\star$ is the natural home for the DeDonder--Weyl
Hamiltonian approach to field theory.  In particular, the isomorphism
\[   J^1 E^\star \cong \Lambda_1^n T^\ast E   \]
makes the first cojet bundle into an $n$-plectic manifold.

In the classical mechanics of a point particle, we can take $\Sigma =
\R$ to represent time and take $M$ to be some manifold representing
space.  Then $E = \Sigma \times M$ is the total space of a
trivial bundle $E \to \Sigma$, and a section of this bundle describes
the path of a particle in space.  The first jet bundle $J^1E$ is the
bundle $\R \times TM$ over the `extended configuration space' $\R
\times M$.  On the other hand, the first cojet bundle $J^1E^{\star}$
is isomorphic as a symplectic manifold to $T^{\ast}(\R \times M)$.
This is the familiar `extended phase space' for a particle in the
space $M$.

For a more relativistic picture, we may instead take $\Sigma = \R$ to
be the parameter space for the path of a particle moving in a manifold
$M$ representing \textit{spacetime}.  As before, $E = \Sigma \times M$
is the the total space of a trivial bundle $E \to \Sigma$, but now a
section of this bundle describes the worldline of a particle in
spacetime.

In Section \ref{strings} we modify this picture a bit further by
letting $\Sigma$ be 2-dimensional, so it represents the parameter
space for a \textit{string} moving in $M$.  We again let $E = \Sigma
\times M$ be the total space of a trivial bundle $E \to \Sigma$.  Now
a section of $E$ describes the worldsheet of a string --- and as we
shall see, the 2-plectic manifold $J^1E^{\star}$ serves as a kind of
`extended phase space' for the string.

\begin{example}
\label{jet.bundle}
As in the previous example, let $\pi \maps E \to \Sigma$ be a fiber bundle 
with $\dim \Sigma = n$.  Let $\Gamma_x (E)$ be the set of
smooth sections of $E$ defined in some neighborhood of the point
$x \in \Sigma$.  Given $\phi \in \Gamma_x (E)$, let $j^1_x \phi$ be the
equivalence class of sections whose first-order Taylor expansion
agrees with the first-order Taylor expansion of $\phi$ at the point $x$.  
The set 
\[  J^1 E= \left \{ j^1_x \phi
~ \vert ~ x \in \Sigma,  \phi \in \Gamma_x(E) \right \} \]
is a manifold.  Moreover, $J^1 E$ is the total space of a
fiber bundle $\pi_J \maps J^1 E \to E$, the {\bf first jet bundle} of 
$E$, where
\[        \pi_J \left(j^1_x \phi \right) = \phi(x).  \]

To see these facts it suffices to work locally, so suppose $E = \Sigma
\times M$.  Let $q^i$ be local coordinates on $\Sigma$ and let $u^a$
be local coordinates on $M$.  These give rise to local coordinates on
$J^1 E$ such that the coordinates for a point $j^1_{x} \phi \in J^1 E$
are $(q^i, u^a, u^a_i)$, where:
\[ 
q^i = q^i(x) , \quad 
u^a = (u^a \circ \phi)(x), \quad
u^a_i = 
\displaystyle{ \frac{\partial u^a \circ \phi}{\partial q^i}(x)} .
\]
The projection $\pi_J$ sends the point with coordinates $(q^i, u^a, u^a_i)$
to the point with coordinates $(q^i, u^a)$, so $\pi_J \maps J^1 E \to E$
is indeed a fiber bundle.

Let $y = (x,u) \in E$. The fiber of $J^1 E$ over $y$ is
\[J^1_y E \cong 
\{ A \maps T_x \Sigma \to T_yE  ~\vert~ d\pi \circ A = 1 \} \]
where $1$ is the identity map on $T_x \Sigma$.  This is not naturally
a vector space, but it is an affine space.  To see this, note that a
difference of two maps $A, A' \maps T_x \Sigma \to T_yE$ lying in
$J^1_y E$ is the same as a linear map from $T_x \Sigma$ to the space
$V_y E$ consisting of vertical vectors at the point $y \in E$.  Thus
$J^1_y E$ is an affine space modeled on the vector space $T^*_x \Sigma
\otimes V_y E$, and $J^1 E$ is a bundle of affine spaces.  (For details,
see Saunders \cite{Saunders}.)

Let $J^1_y E^\star$ be the {\bf affine dual} of $J^1_y E$, that is,
the vector space of affine functions from this affine space
to $\R$.  There is a vector bundle $J^1 E^\star$ over $E$,
the {\bf first cojet bundle} of $E$, whose fiber over $y \in E$ is
$J^1_y E^\star$.  In fact, a volume form on $\Sigma$ determines
a vector bundle isomorphism
\[             J^1 E^\star \cong \Lambda_1^n T^\ast E .\]
With the help of the $n$-plectic structure on $\Lambda_1^n T^\ast E$ 
described in the previous example, this gives an $n$-plectic 
structure on $J^1 E^\star$.

The above isomorphism has been explained by
Cari\~{n}ena, Crampin, Ibort \cite{Carinena-Crampin-Ibort}
and Gotay {\it et al} \cite{GIMM}.  For us it will be
enough to describe it when $E$ is a trivial
bundle over $\Sigma$, say $E = \Sigma \times M$,
and $\Sigma$ is equipped with a volume form, $\vol$.
Using this extra structure, in Example \ref{horizontal}
we constructed a specific isomorphism
\[   \Lambda^n_1 T^*_y E \;\; \cong \;\; \Lambda^n T^*_x \Sigma \; \oplus \;
\Lambda^{n-1} T^*_x \Sigma \otimes T^*_u M \]
where $y = (x,u) \in \Sigma \times M$.
The volume form on $\Sigma$ also determines isomorphisms
\[  
\begin{array}{ccl}
\R &\stackrel{\sim}{\to}& \Lambda^n T^*_x \Sigma \\
                  c &\mapsto & c \, \vol_x  
\end{array}
\] 
and 
\[
\begin{array}{ccl}
   T_x \Sigma   &\stackrel{\sim}{\to} & \Lambda^{n-1} T^*_x \Sigma \\
         v      &\mapsto &    \iota_v \vol_x  .
\end{array}
\]
We thus obtain an isomorphism 
\[   \Lambda^n_1 T^*_y E \;\; \cong \;\; \R \; \oplus \; 
T_x \Sigma \otimes T^*_u M .\]
On the other hand, the trivialization $E = \Sigma \times M$ gives
an isomorphism of affine spaces
\[ J^1_y E \; \; \cong \; \; T^*_x \Sigma \otimes T_u M   \]
which has the side-effect of making $J^1_y E$ into a vector 
space.   When an affine space $V$ happens to be a vector space, 
we have an isomorphism $V^\star \cong \R \oplus V^*$, since 
an affine map to $\R$ is a linear map plus a constant.
So, we obtain
\[   J^1_y E^\star \; \; \cong \; \; \R \; \oplus \; 
T_x \Sigma \otimes T^*_u M .\] 
This gives a specific vector bundle isomorphism
$J^1 E^\star \cong \Lambda_1^n T^\ast E$, as desired. 

It will be useful to see this isomorphism in terms of local coordinates.
We have already described local coordinates
$(q^i, u^a, u^a_i)$ on $J^1 E$.  Taking the affine dual of each fiber,
we obtain local coordinates $(q^i, u^a, P_a^i, P)$ on $J^1 E^\star$.  
We described local coordinates with the same names
on $\Lambda^n_1 T^* E$ in Example \ref{horizontal}.  In terms of
these coordinates, the isomorphism is given simply by
\[  
\begin{array}{ccl}
     J^1 E^\star & \stackrel{\sim}{\to}& \Lambda_1^n T^* E  \\ 
(q^i, u^a, P_a^i, P)  & \mapsto & (q^i, u^a, P_a^i, P) 
\end{array}
\]  
Using this isomorphism to transport the $(n-1)$-form
$i^* \alpha$ given by Eq.\ (\ref{canonical.form}) from
$\Lambda_1^n T^* E$ to $J_1 E^\star$, we obtain this 
differential form:
\begin{eqnarray}
\label{canonical}
   \theta  = P dQ + P^i_a dQ_i^a 
\end{eqnarray}
on $J_1 E^\star$.  Differentiating, it follows that
\begin{eqnarray}
\label{nplectic}
   d\theta  = dP \wedge dQ + dP^i_a \wedge dQ_i^a .
\end{eqnarray}
is an $n$-plectic structure on $J^1 E^\star$.
\end{example}

\section{Poisson Brackets}
\label{Poisson}

Next we quickly review how to generalize Poisson brackets of
observables from symplectic geometry to multisymplectic geometry.
Ordinary Hamiltonian mechanics corresponds to 1-plectic geometry, and
in this case, observables are smooth functions on phase space.  In
$n$-plectic geometry, observables will be smooth $(n-1)$-forms --- but
not all of them, only certain `Hamiltonian' ones:

\begin{definition}
Let $(X,\omega)$ be an $n$-plectic manifold.  An $(n-1)$-form $F$ on $X$ 
is {\bf Hamiltonian} if there exists a vector field $v_F$ on $X$ such that
\begin{eqnarray}
dF= -\ip{F} \omega.
\end{eqnarray}  
We say $v_F$ is the {\bf Hamiltonian vector field} corresponding to $F$. 
The set of Hamiltonian $(n-1)$ forms on a multisymplectic manifold is a 
vector space and is denoted as $\ham$.
\end{definition}

The Hamiltonian vector field $v_F$ is unique if
it exists.  However, except for the familiar case $n = 1$, there may
be $(n-1)$-forms $F$ having no Hamiltonian vector field.  
The reason is that given an $n$-plectic form $\omega$ on $X$, this map: 
\[  \begin{array}{ccl}
     T_x X & \to & \Lambda^n T_{x}^{\ast} X    \\
       v &\mapsto & i_v \omega  
\end{array}
\]
is one-to-one, but not necessarily onto unless $n = 1$.

The following proposition generalizes Liouville's Theorem:

\begin{prop}\label{liouville}If $F \in \ham$, then the Lie derivative
$\Lie{F}{\omega}$ is zero.
\end{prop}

\begin{proof}
Since $\omega$ is closed, $\lie{F}{\omega}= d\ip{F}\omega = 
ddF = 0$.
\end{proof}

We can define a Poisson bracket of Hamiltonian $(n-1)$-forms
in two ways:

\begin{definition}  
\label{hemi-bracket.defn}
Given $F,G \in \ham$, the {\bf hemi-bracket} 
$\Hbrac{F}{G}$ is the $(n-1)$-form given by
\[  \Hbrac{F}{G} = \lie{F}{G}  .\]
\end{definition}

\begin{definition}
\label{semi-bracket.defn}
Given $F,G\in \ham$, the {\bf semi-bracket} $\Sbrac{F}{G}$
is the 
\break
$(n-1)$-form given by 
\[  \Sbrac{F}{G} = \ip{G}\ip{F}\omega .\]
\end{definition}

The two brackets agree in the familiar case
$n = 1$, but in general they differ by an exact form:

\begin{prop} Given $F,G\in \ham$, 
\[
\Hbrac{F}{G}=\Sbrac{F}{G} + d\ip{F}G. 
\label{bracket_relation}
\]
\end{prop}
\begin{proof} Since $L_v = \iota_v d + d \iota_v$,
\[      \Hbrac{F}{G} = 
\lie{F}{G} = 
\ip{F} d G + d \ip{F} G =
-\ip{F}\ip{G} \omega + d \ip{F} G =
\Sbrac{F}{G} + d \ip{F} G .\]
\end{proof}

Both brackets have nice properties:

\begin{prop} \label{hemi-bracket} Let $F,G,H \in \ham$ and
let $v_F,v_G,v_H$ be the respective Hamiltonian
vector fields.  The hemi-bracket $\Hblank$ has the following 
properties:

  \begin{enumerate}
\item The bracket of Hamiltonian forms is Hamiltonian:
  \begin{eqnarray}
    d\Hbrac{F}{G} = -\iota_{[v_F,v_G]} \omega
\label{hemi-closure}
  \end{eqnarray}
so in particular we have 
\[     v_{\Hbrac{F}{G}} = [v_F,v_G]  .\]

\item The bracket is antisymmetric up to an exact form:
  \begin{eqnarray}
    \Hbrac{F}{G} + dS_{F,G} = -\Hbrac{G}{F} 
  \end{eqnarray}
  with $S_{F,G}=-(\ip{F}G + \ip{G}F)$.

\item The bracket satisfies the Jacobi identity:
  \begin{eqnarray}
    \Hbrac{F}{\Hbrac{G}{H}} = 
\Hbrac{\Hbrac{F}{G}}{H} + \Hbrac{G}{\Hbrac{F}{H}}. 
  \end{eqnarray}
  
  \end{enumerate}
\end{prop}

\begin{proof}
1.  If $F,G \in \ham$, then
$d\Hbrac{F}{G}=-\iota_{[v_F,v_G]} \omega - \ip{G}\lie{F}{\omega}$, 
by the identities relating the Lie derivative, exterior derivative, 
and interior product. Prop.\ \ref{liouville} then implies 
the desired result.

2.  Rewriting the Lie derivative in terms of $d$ and $\iota$ gives 
  \[ 
\begin{array}{ccl} 
 \Hbrac{F}{G}+\Hbrac{G}{F} &=& 
\ip{G}\ip{F}\omega + \ip{F}\ip{G}\omega + d(\ip{F}G + \ip{G}F) \\
&=& -dS_{F,G}. 
\end{array}
\]

3.  The definition of the bracket and property 1 give
\[  \begin{array}{ccl}
\Hbrac{\Hbrac{F}{G}}{H} + \Hbrac{G}{\Hbrac{F}{H}} &=&  
  L_{[v_F,v_G]}{H} + \lie{G}{\lie{F}}H \\
  &=& \lie{F}{\lie{G}}H \\
  &=& \Hbrac{F}{\Hbrac{G}{H}}.
\end{array}
\]
\end{proof}

\begin{prop}\label{semi-bracket} Let $F,G,H \in \ham$ and let
$v_F,v_G,v_H$ be the respective Hamiltonian
vector fields.  The semi-bracket $\Sblank$ has the following properties:
  \begin{enumerate}
\item The bracket of Hamiltonian forms is Hamiltonian:
  \begin{eqnarray}
    d\Sbrac{F}{G} = -\iota_{[v_F,v_G]} \omega.
  \end{eqnarray}
so in particular we have 
\[     v_{\Sbrac{F}{G}} = [v_F,v_G]  .\]
\item The bracket is antisymmetric: 
  \begin{eqnarray}
    \Sbrac{F}{G} = -\Sbrac{G}{F}
  \end{eqnarray}

\item The bracket satisfies the Jacobi identity up to an exact form:
\begin{eqnarray}
    \Sbrac{F}{\Sbrac{G}{H}} + dJ_{F,G,H} =
    \Sbrac{\Sbrac{F}{G}}{H} +
    \Sbrac{G}{\Sbrac{F}{H}} 
  \end{eqnarray}
with $J_{F,G,H}=-\ip{F}\ip{G}\ip{H}\omega$.
\end{enumerate}
\end{prop}
\begin{proof} 
1. Prop.\ \ref{bracket_relation} and Prop.\ \ref{hemi-bracket} 
imply $d\Sbrac{F}{G} = -\iota_{\brac{v_F}{v_G}} \omega$.

2. The conclusion follows from the antisymmetry of $\omega$. 

3. First, note that antisymmetry implies 
$\Sbrac{F}{G}=-\ip{F}\ip{G}\omega=\ip{F}dG$. Hence 
\[ \Sbrac{F}{\Sbrac{G}{H}} = \ip{F}d \Sbrac{G}{H}=\ip{F}d \ip{G}dH, \]
\[\Sbrac{\Sbrac{F}{G}}{H}= \ip{\Sbrac{F}{G}} dH = \iota_{[v_{F},v_{G}]} dH, \]
\[\Sbrac{G}{\Sbrac{F}{H}}= \Sbrac{G}{\ip{F}dH}= \ip{G}d \ip{F} dH. \]
The commutator of the Lie derivative and interior product:
\[ \iota_{[v_{F},v_{G}]} = \lie{F}\ip{G} - \ip{G}\lie{F},\]
and Weil's identity: 
\[\lie{F}=d\ip{F}+\ip{F}d, ~\ \lie{G}=d\ip{G}+\ip{G}d,\] 
imply
\begin{eqnarray*}
&&\Sbrac{F}{\Sbrac{G}{H}} - \Sbrac{\Sbrac{F}{G}}{H} - \Sbrac{G}{\Sbrac{F}{H}}\\
&=& \left(- \iota_{[v_{F},v_{G}]}  + \ip{F}d \ip{G} - \ip{G}d \ip{F} \right)dH \\
&=& \left(\ip{G}\lie{F} - \lie{F}\ip{G}+ \ip{F}d \ip{G} - \ip{G}d \ip{F} \right)dH \\
&=& \left(\ip{G}\ip{F}d - d\ip{F}\ip{G} \right)dH\\
&=& - d\ip{F}\ip{G}dH = d\ip{F}\ip{G}\ip{H}\omega = -dJ_{F,G,H}.
\end{eqnarray*} 
\end{proof}

In general, neither the hemi-bracket nor the semi-bracket makes $\ham$
into a Lie algebra, since each satisfies one of the Lie algebra
laws only {\it up to an exact $(n-1)$-form}.  The exception is $n =
1$, the case of ordinary Hamiltonian mechanics.  In this case both
brackets equal the usual Poisson bracket.  In what follows we consider
the case $n = 2$.

\section{Lie 2-Algebras}
\label{lie 2-algebras}

We begin with a quick review of the fully general Lie 2-algebras 
defined by Roytenberg \cite{Roytenberg}.
It will be efficient to work with these using the language 
of chain complexes.  A Lie 2-algebra is a category equipped with 
structures analogous to those of a Lie algebra.  So, to begin with, it is
a `2-vector space': a category where the set of objects and the set of 
morphisms are vector spaces, and all the category operations are linear.
However, one can show \cite{BaezCrans:2004} that a 2-vector space 
is the same as a {\bf 2-term complex}: that is, 
a chain complex of vector spaces that vanishes except in degrees 0 and 1:
\[ L_0 
\stackrel{d}{\leftarrow} L_1 
\stackrel{0}{\leftarrow} 0 
\stackrel{0}{\leftarrow} 0 
\stackrel{0}{\leftarrow} \cdots 
\]
This lets us define a Lie 2-algebra as a 2-term complex
equipped with a bracket operation satisfying the usual Lie algebra
laws `up to coherent chain homotopy'.   

In particular, the bracket
of 0-chains will be skew-symmetric up to a chain homotopy called the
`alternator':
\[        [x,y] + dS_{x,y} = -[y,x]  \]
while the Jacobi identity will hold up to a chain homotopy called the
`Jacobiator':
\[     [x,[y,z]] + dJ_{x,y,z} = [[x,y],z] + [y,[x,z]]  .\]
Furthermore, these chain homotopies need to 
satisfy some laws of their own.  If the alternator vanishes,
a Lie 2-algebra is the same as a 2-term complex made into an
`$L_\infty$-algebra' or `sh Lie algebra' in the sense of 
Stasheff \cite{LadaStasheff:1992}.  Roytenberg introduced 
more general Lie 2-algebras where the alternator does not vanish.

The definitions to come require a few preliminary
explanations.  First, we use the familiar tensor product of chain 
complexes:
\[          (L \otimes M)_i = \bigoplus_{j + k = i} L_j \otimes M_k . \]
With this, the tensor product of 2-term complexes is a 3-term complex. 
Previous work on Lie 2-algebras used a `truncated' tensor product of 
2-term complexes, which gives another 2-term complex \cite{BaezCrans:2004,
Roytenberg}.  But since this makes no difference to anything we do here, 
we shall use the familiar tensor product.

Second, given chain complexes $L$ and $M$, we use
\[     \sigma \maps L \otimes M \to M \otimes L \]
to denote the usual `switch' map with signs included:
\[     \sigma(x \otimes y) = (-1)^{\deg x \deg y} y \otimes x .\]

Third, given 0-chains
$x,y$ and a 1-chain $T$ with $y = x + dT$, we write
\[         T \maps x \to y   .\]
We also write $1 \maps x \to x$ in the case where the 1-chain $T$
vanishes, and write $ST \maps x \to z$ for the 1-chain $S + T$,
where $T \maps x \to y$ and $S \maps x \to z$.  
This notation alludes to how a 2-term chain complex can be 
thought of as a category.  

In this notation, the alternator in a Lie 2-algebra $L$ gives a 1-chain
\[       S_{x,y} \maps [x,y] \to -[y,x]  \]
for every pair of 0-chains $x,y$, and the Jacobiator gives a 1-chain
\[  J_{x,y,z} \maps [x,[y,z]] \to [[x,y],z] + [y,[x,z]]  \]
for every triple of 0-chains $x,y,z$.

\begin{definition}
A {\bf Lie 2-algebra} is a 2-term chain complex of vector spaces
$L = (L_0\stackrel{d}\leftarrow L_1)$ equipped with the following structure:
\begin{itemize}
\item a chain map $\blankbrac\maps L \otimes L\to L$ called the {\bf
bracket};
\item a chain homotopy
\[S\maps \blankbrac \Rightarrow -\blankbrac \circ \sigma\]
called the {\bf alternator};
\item an antisymmetric chain homotopy
\[J\maps \brac{\cdot}{\blankbrac}\Rightarrow \brac{\blankbrac}{\cdot} 
+ \brac{\cdot}{\blankbrac}\circ (\sigma \otimes 1)\]
called the {\bf Jacobiator}.
\end{itemize}
In addition, the following diagrams are required to commute:
$$ \def\objectstyle{\scriptstyle}
  \def\labelstyle{\scriptstyle}
   \xy
   (0,35)*+{[[[w,x],y],z]}="1";
   (-40,20)*+{[[[w,y],x],z] + [[w,[x,y]],z]}="2";
   (40,20)*+{[[[w,x],y],z]}="3";
   (-40,0)*+{[[[w,y],z],x] + [[w,y],[x,z]]}="4'";
   (-40,-4)*+{+ [w,[[x,y],z]] + [[w,z],[x,y]]}="4";
   (40,0)*+{[[[w,x],z],y] + [[w,x],[y,z]]}="5'";
   (-40,-20)*+{[[[w,z],y],x] + [[w,[y,z]],x]}="6'";
   (-40,-24)*+{+ [[w,y], [x,z]] + [w,[[x,y],z]] + [[w,z],[x,y]]}="6";
   (40,-20)*+{[[w,[x,z]],y]}="7'";
   (40,-24)*+{+ [[w,x],[y,z]] + [[[w,z],x],y]}="7";
   (0,-40)*+{[[[w,z],y],x] + [[w,z],[x,y]]  + [[w,y],[x,z]]}="8'";
   (0,-44)*+{+ [w,[[x,z],y]]  + [[w,[y,z]],x] + [w,[x,[y,z]]]}="8";
            (32,-31)*{J_{w,[x,z],y} }; 
            (32,-34.5)*{+  J_{[w,z],x,y} + J_{w,x,[y,z]}};
        {\ar_{[J_{w,x,y},z]}                   "1";"2"};
        {\ar^{1}                               "1";"3"};
        {\ar_{J_{[w,y],x,z} + J_{w,[x,y],z}}   "2";"4'"};
        {\ar_{[J_{w,y,z},x]+1}                 "4";"6'"};
        {\ar^{J_{[w,x],y,z}}                   "3";"5'"};
        {\ar^{[J_{w,x,z},y]+1}                 "5'";"7'"};
        {\ar_{[w,J_{x,y,z}]+1 \; \; }          "6";"8'"};
        {\ar^{}                                "7";"8'"};
\endxy
\\ \\
$$

$$ \def\objectstyle{\scriptstyle}
  \def\labelstyle{\scriptstyle}
   \xy
   (-20,20)*+{[[x,y],z]}="1";
   (20,20)*+{\; -[[y,x],z]}="2";
   (0,0)*+{[x,[y,z]] - [y,[x,z]]}="3";
        {\ar^{[S_{x,y},z]}                   "1";"2"};
        {\ar_{-J_{x,y,z}}          "1";"3"};
		{\ar^{-J_{y,x,z}}        "2";"3"};
\endxy
\\ \\
$$

$$ \def\objectstyle{\scriptstyle}
  \def\labelstyle{\scriptstyle}
   \xy
 (-25,20)*{[x,[y,z]] \;}="1"; 
 (25,20)*{\; -[x,[z,y]]}="2";
 (-25,-20)*{[[x,y],z] + [y,[x,z]] \;}="3"; 
 (25,-20)*{\; -[[x,z],y] - [z,[x,y]]}="4";
{\ar^{[x,S_{y,z}]}   "1";"2"};
{\ar_{S_{[x,y],z}+S_{y,[x,z]}} "3";"4"};
{\ar_{J_{x,y,z}} "1";"3"};
{\ar^{-J_{x,z,y}}  "2";"4"};
\endxy
\\ \\
$$

$$ \def\objectstyle{\scriptstyle}
  \def\labelstyle{\scriptstyle}
   \xy
   (-20,20)*+{[x,[y,z]]}="1";
   (20,20)*+{[x,[y,z]]}="2";
   (0,0)*+{-[[y,z],x]}="3";
        {\ar^{1_{[x,[y,z]]}}                   "1";"2"};
        {\ar_{S_{x,[y,z]}}          "1";"3"};
		{\ar_{-S_{[y,z],x}}        "3";"2"};
\endxy
\\ \\
$$
\end{definition}

\begin{definition}
A Lie 2-algebra for which the Jacobiator is the identity
chain homotopy is called {\bf hemistrict}.  One for which
the alternator is the identity chain homotopy is called {\bf semistrict}.
\end{definition}
\noindent
When the alternator is the identity, the Jacobiator $J_{x,y,z}$ is
antisymmetric as a function of $x,y$ and $z$, so the semistrict
Lie 2-algebras defined here match those of Baez and Crans \cite{BaezCrans:2004}.

Now suppose that $(X,\omega)$ is a 2-plectic manifold.  We shall
construct two Lie 2-algebras associated to $(X,\omega)$: one
hemistrict and one semistrict.  Then we shall prove these are
isomorphic.   
Both these Lie 2-algebras have the same underlying 2-term complex, 
namely:
\[
L \quad = \quad 
\ham  
\stackrel{d}{\leftarrow} \cinf  
\stackrel{0}{\leftarrow} 0 
\stackrel{0}{\leftarrow} 0 
\stackrel{0}{\leftarrow} \cdots 
 \]
where $d$ is the usual exterior derivative of functions.  
To see that this chain complex is well-defined, note that
any exact form is Hamiltonian, with $0$ as its Hamiltonian vector
field.

The hemistrict Lie 2-algebra comes with a bracket called the
{\bf hemi-bracket}:
\[      \Hbrac{\cdot}{\cdot} \maps L \otimes L \to L  .\]
In degree $0$, the hemi-bracket is given as in 
Defn.\ \ref{hemi-bracket.defn}: 
\[  \Hbrac{F}{G} = \lie{F}G.   \]
In degree $1$, it is given by: 
\[  \Hbrac{F}{f} = \lie{F}{f}, \qquad \Hbrac{f}{F} = 0.  \]
In degree $2$, we necessarily have
\[   \Hbrac{f}{g} = 0.    \]
Here $F,G \in \ham$, while $f,g \in \cinf$.

To see that the hemi-bracket is in fact a chain map, it suffices
to check it on hemi-brackets of degree 1:
\[     d\Hbrac{F}{f} = d(\lie{F}f) = \lie{F} df = \Hbrac{F}{df}  \]
and 
\[     d\Hbrac{f}{F} = 0 = \lie{df} F = \Hbrac{df}{F}  \]
since the Hamiltonian vector field corresponding to an exact
1-form is zero.

\begin{theorem}
\label{hemistrict}
If $(X,\omega)$ is a 2-plectic manifold, there is 
a hemistrict Lie 2-algebra $L(X,\omega)_\h$
where:
\begin{itemize}
\item the space of 0-chains is $\ham$,
\item the space of 1-chains is $\cinf$,
\item the differential is the exterior derivative $d \maps \cinf \to \ham$,
\item the bracket is $\Hblank$,
\item the alternator is the bilinear map $S \maps \ham \times \ham \to \cinf$ 
defined by $S_{F,G}= -(\ip{F}G + \ip{G}F)$, and
\item the Jacobiator is the identity, hence given by the trilinear map
$J \maps \ham \times \ham \times \ham \to \cinf$ with 
$J_{F,G,H} = 0$.
\end{itemize}
\end{theorem}

\begin{proof}
That $S$ is a chain homotopy with the right source and target follows
from Prop.\ \ref{hemi-bracket}\comment{.3} and the fact that:
\[   \Hbrac{F}{f} + \Hbrac{g}{G} + S_{F,df} + S_{dg,G} 
= \lie{F}f - \ip{F}df - \ip{G}dg =  -\Hbrac{f}{F} - \Hbrac{G}{g} .\]
Prop.\ \ref{hemi-bracket} \comment{.3} also says that the Jacobi
identity holds.  The following equations then imply
that $J$ is also a chain homotopy with the right source and target:
$$\Hbrac{F}{\Hbrac{G}{f}} = \Hbrac{\Hbrac{F}{G}}{f} +
\Hbrac{G}{\Hbrac{F}{f}} $$
$$\Hbrac{F}{\Hbrac{f}{G}} = \Hbrac{\Hbrac{F}{f}}{G} =
\Hbrac{f}{\Hbrac{F}{G}} = 0$$
$$\Hbrac{f}{\Hbrac{F}{G}} = \Hbrac{\Hbrac{f}{F}}{G} =
\Hbrac{F}{\Hbrac{f}{G}} = 0 .$$ 
So, we just need to check that the Lie 2-algebra axioms hold.  The
first two diagrams commute since each edge is the identity.  The
commutativity of the third diagram is shown as follows:
\[
\begin{array}{rcl}
S_{\Hbrac{F}{G},H}+S_{G,\Hbrac{F}{H}} 
&=& -\iota_{\brac{v_F}{v_G}}H - \ip{H}\Hbrac{F}{G} - \ip{G}\Hbrac{F}{H} 
- \iota_{\brac{v_F}{v_H}}G\\
&=& \Lie{F}{(-\ip{G}H-\ip{H}G)}\\
&=& \Hbrac{F}{S_{G,H}}
\end{array}
\]
The last diagram says that
\[  
S_{F,\Hbrac{G}{H}} - S_{\Hbrac{G}{H},F} = 0,
\]
and this follows from the fact that the alternator
is symmetric: $S_{F,G} = S_{G,F}$.
\end{proof}

Next we make $L$ into a semistrict Lie 2-algebra.  For this,
we use a chain map called the {\bf semi-bracket}:
\[      \Sbrac{\cdot}{\cdot} \maps L \otimes L \to L  .\]
In degree $0$, the semi-bracket is given as in 
Defn.\ \ref{semi-bracket.defn}: 
\[  \Sbrac{F}{G} = \ip{G}\ip{F}\omega .\]
In degrees $1$ and $2$, we set it equal to zero:
\[  \Hbrac{F}{f} = 0, \qquad \Hbrac{f}{F} = 0, \qquad
    \Hbrac{f}{g} = 0.    \]

\begin{theorem}
\label{semistrict}
If $(X,\omega)$ is a 2-plectic manifold, there is a 
semistrict Lie 2-algebra $L(X,\omega)_\s$ where:
\begin{itemize}
\item the space of 0-chains is $\ham$,
\item the space of 1-chains is $\cinf$,
\item the differential is the exterior derivative $d \maps \cinf \to \ham$,
\item the bracket is $\Sblank$,
\item the alternator is the identity, hence given by the bilinear
map $S \maps \ham \times \ham \to \cinf$ with $S_{F,G} = 0$, and
\item the Jacobiator is the trilinear map $J\maps \ham\times \ham\times 
\ham\to \cinf$ defined by $J_{F,G,H} = -\ip{F}\ip{G}\ip{H}\omega$.
\end{itemize}
\end{theorem}

\begin{proof}
We note from Prop\ \ref{semi-bracket}\comment{.2} that the semi-bracket
is antisymmetric.  Since both $S$ and the degree $1$ chain map are zero,
the alternator defined above is a chain homotopy with the right source
and target.  It follows from Prop.\ \ref{semi-bracket}\comment{.3} and that
the Hamiltonian vector field of an exact $1$-form is zero that the Jacobiator
is also a chain homotopy with the desired source and target.
So again, we just need to check that the Lie 2-algebra axioms hold.
The following identities can be checked by simple calculation, and the
commutativity of the first diagram follows:
\[J_{\Sbrac{K}{F},G,H} = 
J_{\Sbrac{H}{K},F,G} - J_{\Sbrac{F}{H},G,K} - \Lie{G}{J_{K,F,H}}\]
\[\Lie{G}{J_{K,F,H}} = 
J_{\Sbrac{G}{K},F,H} + J_{K,\Sbrac{G}{F},H} + J_{K,F,\Sbrac{G}{H}}.\]
Since the Jacobiator is antisymmetric and the alternator is the identity,
the second and third diagrams commute as well.  The fourth diagram
commutes because all the edges are identity morphisms.
\end{proof}

\begin{definition}
Given Lie 2-algebras $L$ and $L'$ with bracket, alternator and 
Jacobiator $\blankbrac$, $S$, $J$ and $\blankbrac^\prime$, $S^\prime$, 
$J^\prime$ respectively, a {\bf homomorphism} from $L$ to $L'$
consists of:
\begin{itemize}
\item a chain map $\phi \maps L \to L'$, and 
\item a chain homotopy $\Phi \maps\blankbrac^\prime\circ (\phi\otimes \phi)
\Rightarrow \phi \circ \blankbrac$
\end{itemize}
such that the following diagrams commute:
\[ \def\objectstyle{\scriptstyle}
  \def\labelstyle{\scriptstyle}
   \xy
 (-25,20)*{[\phi(x),\phi(y)]' \;}="1"; 
 (15,20)*{\; \phi([x,y])}="2";
 (-25,-20)*{-[\phi(y),\phi(x)]' \;}="3"; 
 (15,-20)*{\; -\phi([y,x])}="4";
{\ar^{\Phi_{x,y}}   "1";"2"};
{\ar_{-\Phi_{y,x}} "3";"4"};
{\ar_{S^\prime_{\phi(x),\phi(y)}} "1";"3"};
{\ar^{\phi(S_{x,y})}  "2";"4"};
\endxy
\]

\[ \def\objectstyle{\scriptstyle}
  \def\labelstyle{\scriptstyle}
   \xy
(-25,20)*{[\phi(x),[\phi(y),\phi(z)]^\prime]^\prime \;}="1"; 
(35,20)*{\; [[\phi(x),\phi(y)]^\prime,\phi(z)]^\prime + 
[\phi(y),[\phi(x),\phi(z)]^\prime]^\prime}="2";
 (-25,0)*{[\phi(x),\phi([y,z])]^\prime}="3"; 
(35,0)*{[\phi([y,x]),\phi(z)]^\prime + 
[\phi(y),\phi([x,z])]^\prime}="4";
 (-25,-20)*{\phi([x,[y,z]])\; }="5"; 
(35,-20)*{\; \phi([[x,y],z] + [y,[x,z]])}="6";
{\ar^(.4){J^\prime_{\phi(x),\phi(y),\phi(z)}}   "1";"2"};
{\ar_{[\phi(x),\Phi_{y,z}]^\prime} "1";"3"};
{\ar^{[\Phi_{x,y},\phi(z)]^\prime + 
[\phi(y),\Phi_{x,z}]^\prime} "2";"4"};
{\ar_{\Phi_{x,[y,z]}}  "3";"5"};
{\ar^{\Phi_{[x,y],z} + \Phi_{y,[x,z]}} "4";"6"};
{\ar_{\phi(J_{x,y,z})} "5";"6"};
\endxy
\]
\end{definition}

Roytenberg explains how to compose Lie 2-algebra homomorphisms 
\cite{Roytenberg}, and we say a Lie 2-algebra homomorphism with an
inverse is an {\bf isomorphism}.

\begin{theorem}
\label{isomorphism}
$L(X,\omega)_\h$ and $L(X,\omega)_\s$ are isomorphic as Lie 2-algebras.
\end{theorem}

\begin{proof}
We show that the identity chain maps with appropriate chain homotopies
define Lie 2-algebra homomorphisms and that their composites are the
respective identity homomorphisms.  There is a homomorphism $\phi \maps
L(X,\omega)_\h \to L(X,\omega)_\s$ with the identity chain map and the
chain homotopy given by $\Phi_{F,G} = \ip{F}G$.  That this is a chain
homotopy follows from the bracket relation $\Sbrac{F}{G} + d\left(\ip{F}G\right)
= \Hbrac{F}{G}$ noted in Prop.\ \ref{bracket_relation} together
with the equations
$$ \Sbrac{F}{f} + \ip{F}df = \Hbrac{F}{f}, \qquad
\Sbrac{f}{F} = \Hbrac{f}{F} = \ip{df}F = 0. $$
We check that the two diagrams in the definition 
of a Lie 2-algebra homomorphism commute.
Noting that the chain map is the identity, the commutativity of the
first diagram is easily checked by recalling that $S_{F,G} = -(\ip{G}F
+ \ip{F}G)$ and that $S'_{F,G}$ is the identity.  Noting that any edge
given by the bracket for $L(X,\omega)_\s$ in degree $1$ is the identity 
and that $J_{F,G,H}$ is the identity, to check the commutativity of the second
diagram we only need to perform the following calculation:
\begin{eqnarray*}
&&J'_{F,G,H} + \Phi_{\Hbrac{F}{G},H} + \Phi_{G,\Hbrac{F}{H}} -
\Phi_{F,\Hbrac{G}{H}} \\ 
&=& -\ip{F}\ip{G}\ip{H}\omega + \iota_{\brac{v_F}{v_G}}H + \ip{G}\Lie{F}{H} - \ip{F}\Lie{G}{H}\\
&=& \ip{F}\Lie{G}{H} - \ip{F}d\ip{G}H + \iota_{\brac{v_F}{v_G}}H + 
\ip{G}\Lie{F}{H} - \ip{F}\Lie{G}{H}\\
&=& -\ip{F}d\ip{G}H + \iota_{\brac{v_F}{v_G}}H + \ip{G}\Lie{F}{H}\\
&=& -\ip{F}d\ip{G}H + \Lie{F}{\ip{G}H} - \ip{G}\Lie{F}{H} + \ip{G}\Lie{F}{H}\\
&=& -\ip{F}d\ip{G}H + \Lie{F}{\ip{G}H}\\
&=& d\ip{F}\ip{G}H - \Lie{F}\ip{G}H + \Lie{F}{\ip{G}H}\\
&=& d\ip{F}\ip{G}H\\
&=& 0 .
\end{eqnarray*}
\end{proof}

\section{The Classical Bosonic String} 
\label{strings}

The bosonic string is a theory of maps $\phi \maps \Sigma \to M$ where
$\Sigma$ is a surface and $M$ is some manifold representing spacetime.
For simplicity we will only consider the case where $\Sigma$ is the
cylinder $\CYL$ and $M$ is $d$-dimensional Minkowski spacetime, $\ST$.
A solution of the classical bosonic string is then a map $\phi \maps
\Sigma \to M$ which is a critical point of the area subject to certain
boundary conditions.

Equivalently, by exploiting symmetries
in the variational problem, one can describe solutions $\phi$ by equipping
$\CYL$ with its standard Minkowski metric and then solving the
$1+1$ dimensional field theory specified by the Lagrangian density
\begin{eqnarray*}
\ell=\frac{1}{2} g^{ij}\eta_{ab} \frac{\partial
  \phi^{a}}{\partial q^i}\frac{\partial \phi^{b}}{\partial q^j}.
\end{eqnarray*}
Here $q^i$ $(i = 0,1)$ are standard coordinates on $\CYL$ and
$g=\mathrm{diag}(1,-1)$ is the Minkowski metric on $\CYL$, while
$\phi^a$ are the coordinates of the map $\phi$ in $\ST$ and $\eta =
\mathrm{diag}(1,-1,\hdots,-1)$ is the Minkowski metric on $\ST$.  We
use the Einstein summation convention to sum over repeated indices.
The corresponding Euler--Lagrange equation is just a version of the
wave equation:
\begin{eqnarray*}
g^{ij}\partial_{i} \partial_{j} \phi^a =0.
\label{eom}
\end{eqnarray*}
We next describe this theory using multisymplectic geometry following
the approach of H\'{e}lein \cite{Helein}.  (The work of Gotay {\it et
al} \cite{GIMM} focuses instead on the Polyakov approach, where the
metric on $\Sigma$ is taken as an independent variable.)

The space $E=\Sigma \times M$ can be thought of as a trivial bundle
over $\Sigma$, and the graph of a function $\phi \maps \Sigma \to M$
is a smooth section of $E$.  We write the coordinates of a
point $(x,u)\in E$ as $\left(q^i,u^a \right)$.
Let $J^1 E \to E$ be the first jet bundle of $E$.
As explained in Example \ref{jet.bundle}, since $E$ is trivial
we may regard $J^1 E$ as a vector bundle whose
fiber over $(x,u)\in E$ is $T^*_x \Sigma \otimes T_u M$.  
The Lagrangian density for
the string can be defined as a smooth function on $J^1 E$:
\begin{eqnarray*}
\ell=\frac{1}{2} g^{ij}\eta_{ab}u^{a}_{i}u^{b}_{j},
\end{eqnarray*}
which depends in this example only on the fiber coordinates
$u^a_{i}$.  

Let $J^1 E^* \to E$ be the vector bundle dual to $J^1 E$.  
The fiber of $J^1 E^*$ over $(x,u) \in E$ is $T_x \Sigma \otimes 
T_u^* M$.  From the Lagrangian $\ell \maps J^1E \to \R$, the `de 
Donder--Weyl Hamiltonian' $h \maps J^1 E^* \to \R$ can be
constructed via a Legendre transform.   It is given as follows:
\begin{eqnarray*}
h&=& p^{i}_{a}u^{a}_{i}- \ell \\
&=&\frac{1}{2} \eta^{ab}g_{ij}p_{a}^{i}p_{b}^{j},
\end{eqnarray*}
where $u^a_{i}$ are defined implicitly by
$p_a^{i}=\partial \ell / \partial u^{a}_{i}$,
and $p_a^{i}$ are coordinates on the
fiber $T^{\ast}_{u}M \otimes T_{x}\Sigma$.
Note that $h$ differs from the standard (non-covariant) Hamiltonian density
$\varepsilon$ for a field theory: 
\begin{eqnarray*}
\varepsilon&=& p^{0}_{a}u^{a}_{0} - \ell\\
&=& \frac{1}{2} \eta^{ab}\left
(p_{a}^{0}p_{b}^{0} + p_{a}^{1}p_{b}^{1} \right).
\end{eqnarray*}

Let $\phi$ be a section of $E$ and let $\pi$ be a smooth section of
$J^1 E^*$ restricted to $\phi(\Sigma)$ with fiber coordinates
$\pi_{a}^{i}$.  It is then straightforward to show
that $\phi$ is a solution of the Euler--Lagrange equations 
if and only if $\phi$ and $\pi$ satisfy the following system of
equations:
\begin{eqnarray}
\frac{\partial \pi^{i}_{a}}{\partial q^i} &=& 
  -  \left.\frac{\partial h}{\partial
  u^{a}} \right \vert_{u=\phi,p=\pi}  \\
  \frac{\partial \phi^{a}}{\partial q^i} &=& \left.\frac{\partial h}{\partial
  p^i_{a}} \right \vert_{u=\phi,p=\pi}. \label{ham_eqs}
\end{eqnarray}
This system of equations is a generalization of Hamilton's equations
for a classical point particle.

As explained in Example \ref{jet.bundle} and the preceding discussion,
the extended phase space for the string is the first cojet bundle 
$J^1 E^\star$, and this space is equipped with a canonical 2-form $\theta$ 
whose exterior derivative $\omega = d \theta$ is a 2-plectic structure.
Using the isomorphism
\[         J^1 E^\star \cong J^1 E^* \times \R ,\]
a point in $J^1 E^\star$ gets coordinates $(q^i,u^a,p^{i}_{a},e)$. 
In terms of these coordinates, 
\[
\theta= e ~dq^{0} \wedge dq^{1} +
\left(p_{a}^{0} du^{a} \wedge dq^{1} - 
p_{a}^{1} du^{a} \wedge dq^{0} \right) .
\]
The 2-plectic structure on $J^1 E^\star$ is thus
\[   \omega = 
de \wedge dq^0 \wedge dq^{1} +
\left(dp_a^0 \wedge du^a \wedge dq^{1} - 
dp_a^1 \wedge du^a \wedge dq^{0} \right) .  \]
So, the variable $e$ may be considered as `canonically conjugate' 
to the area form $dq^{0} \wedge dq^{1}$.  

As before, let $\phi$ be a section of $E$ and let $\pi$ be a smooth section of
$J^1 E^*$ restricted to $\phi(\Sigma)$.  
Consider the submanifold $S \subset J^1 E^\star$ with coordinates:
\begin{eqnarray*}
(q^i,\phi^{a}(q^j),\pi_{a}^{i}(q^j),-h).
\end{eqnarray*}
Note that $S$ is constructed from $\phi$, $\pi$ and
from the constraint $e + h=0$. This constraint is analogous to the one
that is used in finding constant energy solutions in the extended
phase space approach to classical mechanics.  At each point in
$S$, a tangent bivector $v=v_{0} \wedge v_{1}$ can be defined as
\[
\begin{array}{ccl}
 v_{0} &=& {\displaystyle \frac{\partial}{\partial q^{0}} + 
\frac{\partial \phi^a}{\partial q^{0}}\frac{\partial}{\partial u^a} + 
\frac{\partial \pi_{a}^{i}}{\partial q^{0}}
\frac{\partial}{\partial p_{a}^{i}}} \\   \\
 v_{1} &=& {\displaystyle \frac{\partial}{\partial q^{1}} + 
\frac{\partial \phi^a}{\partial q^{1}}
\frac{\partial}{\partial u^a} + 
\frac{\partial \pi_{a}^{i}}{\partial q^{1}}
\frac{\partial}{\partial p_{a}^{i}}. } 
\end{array}
\]
Explicit computation reveals that the submanifold $S$ is
generated by solutions to Hamilton's equations if and only if
\[
\omega(v_{0},v_{1},\cdot)=0. 
\]

Quite generally, infinitesimal symmetries of the 2-form $\theta$ give
rise to Hamiltonian 1-forms that generate these symmetries.  For
example, symmetry under time evolution lets us define a Hamiltonian.
Consider the Lie derivative of $\theta$ along the coordinate vector field
$\partial/\partial q^{0}$:
\[ 
\begin{array}{ccl}
L_{\partial/\partial q^{0}} \theta &=&
{\displaystyle d\iota_{\partial/\partial q^{0}}\theta + 
\iota_{\partial/\partial q^{0}}\omega } \\  
&=& d \left( e~dq^{1} + p_{a}^{1}du^{a} \right)  - \left(de \wedge dq^{1} + 
dp_{a}^{1} \wedge du^{a} \right)  \\  
&=&0. \label{lie_theta}
\end{array}
\]
Hence $\theta$ is invariant with respect to infinitesimal
displacements along the $q^{0}$ coordinate.  If we define
a 1-form $H$ by
\[
\begin{array}{ccl}
H&=&{\displaystyle -\iota_{\partial/ \partial q^{0}} \theta} \\ 
&=&-e~dq^{1} - p_{a}^{1} du^{a}.
\end{array}
\]
then $dH = \iota_{\partial/\partial q^0} \omega$.  Hence $H$ is
a Hamiltonian 1-form, and the Hamiltonian vector
field $v_H$ describes time evolution.

One may wonder how this Hamiltonian 1-form $H$ is related to the
usual concept of energy.  To understand this,
consider the solution submanifold $S$ as defined
above. Let $S_{\tau} \subset S$ be a 1-dimensional curve
on $S$ at constant `time' $q^{0}=\tau$. Denote the restriction 
of $H$ to $S_{\tau}$ as $H_{\tau}$.  A computation yields
\begin{eqnarray*}
H_{\tau}&=& h~dq^{1} - \pi_{a}^{1} d \phi^a\\
&=&\frac{1}{2} \eta^{ab}g_{ij}
\pi_{a}^{i}\pi_{b}^{j}d q^{1}
- \pi_{a}^{1} d \phi^a.
\end{eqnarray*}
On $S_{\tau}$, $dq^{0}=0$. Hence $d\phi^{a} = \frac{\partial
\phi^{a}}{\partial q^1} dq^1 $.  Since $\phi$ satisfies
Eq.\ \eqref{ham_eqs}, we also have:
\begin{eqnarray*}
\pi^{0}_{a} = \eta_{ab}\frac{ \partial \phi^{b}}{\partial q^{0}},\\
\pi^{1}_{a} = - \eta_{ab}\frac{ \partial \phi^{b}}{\partial q^{1}}.
\end{eqnarray*}
The expression for $H_{\tau}$ thus becomes:
\begin{eqnarray*}
H_{\tau} &=& \frac{1}{2} \eta^{ab}\left
(\pi_{a}^{0}\pi_{b}^{0} + \pi_{a}^{1}\pi_{b}^{1}
\right) dq^{1}\\ 
&=&\frac{1}{2} \eta_{ab}\left (\frac{\partial
\phi^{a}}{\partial q^{0}} \frac{\partial \phi^{b}}{\partial q^{0}} + \frac{\partial
\phi^{a}}{\partial q^{1}}\frac{\partial \phi^{b}}{\partial q^{1}}
\right) dq^{1} \\
&=& \varepsilon ~ d q^{1}.
\end{eqnarray*}
Hence $H_{\tau}$ is the Hamiltonian 1-form that corresponds to the
energy density of the string at $\tau$, and the total energy of
the string at $q^{0}=\tau$ is simply:
\begin{eqnarray*}
\int_{S_{\tau}} H_{\tau}.    
\end{eqnarray*}
So, the usual concept of energy is compatible with the concept of
energy as a Hamiltonian 1-form in the Lie 2-algebra of observables
for the string.

We next consider a scenario in which the string is coupled to a 
a $B$ field.  We fix a 2-form $B$ on $M$.  By pulling back $dB$ along the 
projection $p \maps J^1 E^\star \to M$ and adding it to the 2-plectic form 
$\omega$, we obtain a modified 2-plectic form $\tilde{\omega}$ on $J^1 E^\star$: 
\[ \tilde{\omega} = \omega + p^* dB. \]
In coordinates:
\begin{eqnarray*}
p^* dB= d \left(B_{bc} \, du^b \wedge du^c \right)=
\frac{\partial B_{bc}}{\partial u^{a}} \, du^{a} \wedge du^{b} \wedge du^{c}.
\end{eqnarray*} 
It is straightforward to show that $\tilde{\omega}$ is indeed
2-plectic.

We now determine the equations of motion for the string coming from
the modified 2-plectic structure $\tilde{\omega}$. 
As before, we consider the submanifold $S$ defined above.
We emphasize that we have not changed $h$: it is still the
de Donder--Weyl Hamiltonian for the free string.  Requiring 
$\tilde{\omega}(v_{0},v_{1},\cdot)=0$ implies
\[
\omega(v_{0},v_{1},\cdot)+
p^* dB(v_{0},v_{1},\cdot)=0.
\]
Let 
\[ 
 J^{bc}= \frac{\partial \phi^{b}}{\partial q^0}\frac{\partial
  \phi^{c}}{\partial q^1}
 -\frac{\partial \phi^{b}}{\partial q^1}\frac{\partial
  \phi^{c}}{\partial q^0}
\] 
and 
\[ 
F_{bcd}=
 \frac{\partial B_{cd}}{\partial u^b} + 
 \frac{\partial B_{db}}{\partial u^c} +
 \frac{\partial B_{bc}}{\partial u^d}. 
\]
It follows that
\[
(p^* dB)(v_{0},v_{1},\cdot) = J^{bc} F_{bcd} du^{d}, 
\]
which implies that $\phi$ obeys the following equations:
\[
g^{ij}\partial_{i}\partial_{j}\phi^{a}=
\eta^{ad}J^{bc} F_{bcd}. 
\]
These equations, familiar from the work of Kalb and Ramond 
\cite{Kalb-Ramond}, are precisely the Euler--Lagrange equations 
derived from a Lagrangian density $\tilde{\ell}$ that includes 
a $B$ field term:
\[
\tilde{\ell}= \ell + J^{ab}B_{ab}.
\]
So, adding the pullback of $dB$ to $\omega$ modifies the 2-plectic
structure in precisely the right way to give the correct equations of
motion for a string coupled to a $B$ field.  This generalizes the
usual story for point particles coupled to electromagnetism \cite{GS}.

\section{Conclusions}
\label{conclusions}

The work presented here raises many questions.  Here are four
obvious ones:

\begin{itemize}
\item Does an $n$-plectic manifold give rise to a Lie $n$-algebra when
$n > 2$?  There is not yet a definition of weak or hemistrict Lie
$n$-algebras for $n > 2$, but a semistrict Lie $n$-algebra is just an
$n$-term chain complex equipped with the structure of an
$L_\infty$-algebra.  So, it would be easiest to start by considering a
generalization of the semi-bracket, and see if this can be used to
construct a semistrict Lie $n$-algebra.

\item Does the Lie 2-algebra of observables in 2-plectic geometry
extend to something like a Poisson algebra?  It is far from clear how
to define a product for Hamiltonian 1-forms, and the usual product of
a Hamiltonian 1-form and a smooth function is not Hamiltonian.

\item
The based loop space $\Omega X$ of a manifold $X$ equipped with a
closed $(n+1)$-form $\omega$ is an infinite-dimensional manifold 
equipped with a closed $n$-form $\eta$ defined `by transgression'
as follows: 
\[         \eta(v_1, \dots, v_n) = 
\int_0^{2\pi} 
\omega(\gamma'(\sigma),v_1(\gamma(\sigma)), \dots,
v_n(\gamma(\sigma)) \; d\sigma \]
where $v_i$ are tangent vectors at the loop $\gamma \in \Omega X$ and
$v_i(\gamma(\sigma))$ are the corresponding tangent vectors
at the point $\gamma(\sigma) \in X$.  Even when $\omega$ 
is $n$-plectic, $\eta$ is rarely $(n-1)$-plectic.  However
when $X = G$ is a compact simple Lie group equipped with the
2-plectic structure of Example \ref{lie.group},
$\eta$ becomes symplectic after adding an exact form.  
The interplay between the 2-plectic structure on $G$ and the 
symplectic structure on $\Omega G$ plays an important role in
the theory relating the Wess--Zumino--Witten model, central
extensions of the loop group $\Omega G$, gerbes on $G$ and
the string 2-groups ${\rm String}_k(G)$ \cite{BCSS}.  It would be
nice to have a more general theory whereby the loop space of
an $n$-plectic manifold became an $(n-1)$-plectic manifold.

\item When a symplectic structure $\omega$ on a manifold $X$ defines 
an integral class in $H^2(X,\R)$, there is a $\U(1)$ bundle over $X$
equipped with a connection whose curvature is $\omega$.  As mentioned
in the Introduction, this plays a fundamental role in the geometric 
quantization of $X$.  Similarly,
when a 2-plectic structure $\omega$ on a manifold $X$ defines an
integral class in $H^3(X,\R)$, there is a $\U(1)$ gerbe over $X$ equipped
with a connection whose curvature is $\omega$ \cite{Brylinski}.
Is there an analogue of geometric quantization that applies in
this case?  

Following the ideas of Freed \cite{Freed:1994}, we might hope that
geometrically quantizing this gerbe will give a `2-Hilbert space' of
states.  However, Freed's work only treats Schr\"odinger quantization,
and that only in the special case where the resulting 2-Hilbert space
is finite-dimensional.  Finite-dimensional 2-Hilbert spaces are by now
well-understood \cite{HDA2}, but the infinite-dimensional ones are
still being developed \cite{BBFW,Yetter}.  Geometric quantization for
gerbes is an even greater challenge.  However, we expect the problem
of geometrically quantizing a $\U(1)$ gerbe on $X$ to be closely
related to the better-understood problem of geometrically quantizing
the corresponding $\U(1)$ bundle on the loop space of $X$.
\end{itemize}

\section{Acknowledgments}
We thank Urs Schreiber, Allen Knutson and Dmitry Roytenberg
for corrections and helpful conversations.


\begin{thebibliography}{99}

\bibitem{AschieriCantiniJurco:2003}
P.~Aschieri, L.~Cantini, and B.~Jur{\v c}o, Nonabelian bundle gerbes,
their differential geometry and gauge theory, 
{\sl Comm.\ Math.\ Phys.\ }{\bf 254} (2005), 367--400. 
Also available as \href{http://arxiv.org/abs/hep-th/0312154}{\texttt arXiv:hep-th/0312154}.

\bibitem{AschieriJurco:2004}
P.~Aschieri and B.~Jur{\v c}o, Gerbes, {M5}-brane anomalies and {$E_8$}
gauge theory, {\sl JHEP} {\bf 10} (2004) 068. 
Also available as \href{http://arxiv.org/abs/hep-th/0312154}{\texttt arXiv:hep-th/0409200}.

\bibitem{HDA2} J.\ Baez, Higher-dimensional algebra {II}: 2-Hilbert
spaces, \textit{Adv. Math.} \textbf{127} (1997), 125--189. Also available as
\href{http://arxiv.org/abs/q-alg/9609018}{\texttt arXiv:q-alg/9609018}.

\bibitem{BBFW} J.\ Baez, A.\ Baratin, L.\ Freidel and D.\ Wise,
Representations of 2-groups on higher Hilbert spaces, manuscript
in preparation.

\bibitem{BaezCrans:2004}
J.\ Baez and A.\ Crans, Higher-dimensional algebra VI: Lie 2-algebras, 
\href{http://www.tac.mta.ca/tac/volumes/12/14/12-14abs.html}{\sl TAC}
{\bf 12}, (2004), 492--528.  Also available as 
\href{http://arxiv.org/abs/math/0307263}{\texttt arXiv:math/0307263}.

\bibitem{BCSS} J.\ Baez, A.\ Crans, D.\ Stevenson and U.\ Schreiber, 
  From loop groups to 2-groups, \href{http://intlpress.com/HHA/v9/n2/a4/}
{\sl HHA} {\bf 9} (2007), 101--135.  Also available as 
\href{http://arxiv.org/abs/math/0504123}{\texttt arXiv:math/0504123}.

\bibitem{BaezLauda:2003}
J.\ Baez and A.\ Lauda, Higher-dimensional algebra V: 2-groups,  
\href{http://www.tac.mta.ca/tac/volumes/12/15/12-15abs.html}{\sl TAC} 
{\bf 12} (2004), 423--491.  Also available as 
\href{http://arxiv.org/abs/math.QA/0307200}{\texttt arXiv:math.QA/0307200}.

\bibitem{BaezSchreiber:2005} J.\ Baez and U.\ Schreiber, Higher
gauge theory, in {\sl Categories in Algebra, Geometry and Mathematical 
Physics}, eds.\ A.\ Davydov {\it et al}, {\sl Contemp. Math.} 
{\bf 431}, AMS, Providence, Rhode Island, 2007, pp.\ 7--30.
Also available as \href{http://arxiv.org/abs/math/0511710}{\texttt arXiv:math/0511710}.

\bibitem{Bartels:2004}
T.~Bartels, Higher gauge theory: 2-bundles, available as \hfill \break
\href{http://arxiv.org/abs/math/0410328}
{\texttt arXiv:math/0410328}.

\bibitem{BreenMessing:2001}
L.\ Breen and W.\ Messing, Differential geometry of gerbes,
{\sl Adv.\ Math.\ }{\bf 198} (2005), 732--846.
Also available as \href{http://arxiv.org/abs/math/0106083}{arXiv:math/0106083}.

\bibitem{Brylinski}
J.-L.~Brylinski, {\sl Loop Spaces, Characteristic Classes and Geometric
Quantization}, Birkhauser, Boston, 1993.

\bibitem{Cantrijn:1999} F.\ Cantrijn, A.\ Ibort, and M. De Leon, On
the geometry of multisymplectic manifolds, {\sl J.\ Austral.\ Math.\
Soc.\ (Series A)} \textbf{66} (1999), 303--330.

\bibitem{Cartin} D.\ Cartin, Generalized symplectic manifolds,
available as \hfill \break
\href{http://arxiv.org/abs/dg-ga/9710027}{\texttt arXiv:dg-ga/9710027}.

\bibitem{Carinena-Crampin-Ibort} J.\ F.\ Cari\~{n}ena, M.\ Crampin, 
and L.\ A.\ Ibort, On the multisymplectic formalism for first order
field theories, \textsl{Diff.\ Geom.\ Appl.\ } \textbf{1} (1991), 345--374.

\bibitem{DeDonder} T.\ DeDonder, \textit{Theorie Invariantive du 
Calcul des Variations}, Gauthier--Villars, Paris, 1935.

\bibitem{Diaconescu-Moore-Freed} E.\ Diaconescu, G.\ Moore, and 
D.\ Freed, The $M$-theory 3-form and $E\sb 8$ gauge theory, in 
\textsl{Elliptic cohomology}, eds.\ H.\ R.\ Miller and D.\ C.\ 
Ravenel, London Math. Soc. Lecture Note Ser., 342,
Cambridge Univ. Press, Cambridge, 2007, pp.\ 44--88.
Also available as
\href{http://arxiv.org/abs/hep-th/0312069}
{\texttt arXiv:hep-th/0312069}.

\bibitem{Freed:1994} D.\ Freed, Higher algebraic structures and
quantization, \textsl{Comm.\ Math.\ Phys.} \textbf{159} (1994),
343--398.  Also available as
\href{http://arxiv.org/abs/hep-th/9212115}{\texttt arXiv:hep-th/9212115}.

\bibitem{Freed:2001} D.\ Freed, Dirac charge quantization and
generalized differential cohomology, in \textsl{Surveys in
differential geometry VII}, ed.\ S.\ -T.\ Yau, International Press,
Somerville, Massachusetts, 2000, pp.\ 129--194. Also available as
\href{http://arxiv.org/abs/hep-th/0011220}{\texttt
arXiv:hep-th/0011220}.

\bibitem{Freed-Witten} D.\ Freed and E.\ Witten, Anomalies in string theory 
with D-branes,  \textsl{Asian J.\ Math.}  \textbf{3} (1999), 819--851. Also
available as \href{http://arxiv.org/abs/hep-th/9907189}
{\texttt arXiv:hep-th/9907189}.

\bibitem{GIMM} M.\ Gotay, J.\ Isenberg, J.\ Marsden, and R.\ Montgomery,
Momentum maps and classical relativistic fields. Part I: covariant 
field theory, available as
\href{http://arxiv.org/abs/physics/9801019}{\texttt arXiv:physics/9801019}.

\bibitem{GS} V.\ Guillemin and S.\ Sternberg, {\it Symplectic
Techniques in Physics}, Cambridge U.\ Press, Cambridge, 1984.

\bibitem{Helein} F.\ H\'{e}lein, Hamiltonian formalisms for
multidimensional calculus of variations and perturbation theory, in
\textsl{Noncompact Problems at the Intersection of Geometry}, eds.\
A.\ Bahri \textit{et al}, AMS, Providence, Rhode Island, 2001, pp.\
127--148. Also available as
\href{http://arxiv.org/abs/math-ph/0212036}{\texttt
arXiv:math-ph/0212036}.

\bibitem{Helein-Kouneiher} F.\ H\'{e}lein and J.\ Kouneiher, The
notion of observable in the covariant Hamiltonian formalism for the
calculus of variations with several variables, \textsl{Adv.\ Theor.\
Math.\ Phys.} \textbf{8} (2004), 735--777. Also available as
\href{http://arxiv.org/abs/math-ph/0401047} {\texttt
arXiv:math-ph/0401047}.

\bibitem{Ibort:2000} A.\ Ibort, Multisymplectic geometry: generic and
exceptional, in \textsl{Proceedings of the IX Fall Workshop on
Geometry and Physics, Vilanova i la Geltr\'u, 2000}, eds.\ X.\
Gr\'{a}cia \textit{et al}, Publicaciones de la RSME vol.\ 3, Real
Sociedad Matem\'atica Espa\~nola, Madrid, 2001, pp.\ 79--88.

\bibitem{Kalb-Ramond}
M.\ Kalb and P.\ Ramond, Classical direct interstring action,
{\sl Phys.\ Rev.\ D.\ } {\bf 9} (1974), 2273--2284.

\bibitem{Kijowski} J.\ Kijowski, A finite-dimensional canonical
formalism in the classical field theory, \textit{Commun.\ Math.\
Phys.} \textbf{30} (1973), 99--128.

\bibitem{LadaStasheff:1992}
T.\ Lada and J.\ Stasheff, Introduction to sh Lie
algebras for physicists, {\sl Int.\ Jour.\ Theor.\ Phys.}
\textbf{32} (7) (1993), 1087--1103.  Also available as
\href{http://arxiv.org/abs/hep-th/9209099}{\texttt
arXiv:hep-th/9209099}.

\bibitem{Rovelli} C.\ Rovelli, Covariant Hamiltonian formalism for
field theory: Hamilton-Jacobi equation on the space $G$, available as
\href{http://lanl.arxiv.org/abs/gr-qc/0207043}{\texttt arXiv:gr-qc/0207043}.

\bibitem{Roytenberg} D.\ Roytenberg, On weak Lie 2-algebras, available
as \href{http://arxiv.org/abs/0712.3461}{\texttt arXiv:0712.3461}.

\bibitem{Saunders} D.\ J.\ Saunders, \textsl{The Geometry of Jet
Bundles}, London Math.\ Soc.\ Lecture Note Ser. {\bf 142},
Cambridge U.\ Press, Cambridge, 1989.

\bibitem{SSS} H.\ Sati, U.\ Schreiber and J.\ Stasheff, 
$L_\infty$-algebra connections and applications to String-
and Chern--Simons $n$-transport, available as
\href{http://arxiv.org/abs/0801.3480}{\texttt arXiv:0801.3480}.

\bibitem{SSS2} H.\ Sati, U.\ Schreiber and J.\ Stasheff, 
Fivebrane structures, available as
\href{http://arxiv.org/abs/0805.0564}{\texttt arXiv:0805.0564}.

\bibitem{Schreiber:2005}
U.\ Schreiber, {\sl From Loop Space Mechanics to Nonabelian Strings},
Ph.D.\ thesis, Universit\"at Duisburg-Essen, 2005.  Also available as
\href{http://arxiv.org/abs/hep-th/0509163}{\texttt arXiv:hep-th/0509163}.

\bibitem{Weyl} H.\ Weyl, Geodesic fields in the calculus of variation
for multiple integrals, \textsl{Ann.\ Math.} \textbf{36} (1935),
607--629.

\bibitem{Yetter} D.\ Yetter, Measurable categories, {\sl Appl.\ Cat.\
Str.\ } \textbf{13} (2005), 469--500.  Also available as
\href{http://arxiv.org/abs/math.CT/0309185}{arXiv:math.CT/0309185}.

\bibitem{Zunger} Y.\ Zunger, $p$-Gerbes and extended objects in 
string theory, \hfill \break
\href{http://arxiv.org/abs/hep-th/0002074}{\texttt arXiv:hep-th/0002074}.

\end{thebibliography}
\end{document}